\documentclass[11pt]{amsart}
\pdfoutput=1
\usepackage{amssymb,latexsym}
\usepackage{amstext,amsthm,amscd}
\usepackage{amsmath}
\usepackage{amsgen,amsfonts,amsbsy}

\theoremstyle{plain}
\newtheorem{theorem}{Theorem}[section]
\newtheorem{proposition}[theorem]{Proposition}
\newtheorem{lemma}[theorem]{Lemma}
\newtheorem{corollary}[theorem]{Corollary}

\theoremstyle{definition}
\newtheorem{definition}[theorem]{Definition}

\newcommand{\Z}{\mathbb{Z}}
\newcommand{\N}{\mathbb{N}}

\newcommand{\R}{\mathbb{R}}
\newcommand{\C}{\mathbb{C}}
\newcommand{\set}[2]{\{#1|\ #2\}}
\newcommand{\sub}{\subseteq}
\newcommand{\gen}[1]{\left\langle #1\right\rangle}
\newcommand{\M}{\mathrm{M}}
\newcommand{\Tr}{\mathrm{Tr}}
\newcommand{\GL}{\mathrm{GL}}
\newcommand{\Ad}{\mathrm{Ad}}
\newcommand{\Inn}{\mathrm{Int}}
\newcommand{\lcm}{\mathrm{lcm}}
\newcommand{\Aut}{\mathrm{Aut}}
\newcommand{\diag}{\mathrm{diag}}

\newcommand{\SL}{\mathrm{SL}}
\newcommand{\Sp}{\mathrm{Sp}}
\newcommand{\gl}{\mathrm{gl}}
\renewcommand{\sl}{\mathrm{sl}}
\renewcommand{\P}{\mathcal{P}}
\renewcommand{\H}{\mathcal{H}}

\newenvironment{mat2}[4]{\left(\begin{array}{cc} #1&#2\\#3&#4 \end{array}\right)}{}
\newenvironment{matt2}[4]{\left[\begin{array}{cc} #1&#2\\#3&#4 \end{array}\right]}{}
\newenvironment{smat2}[4]{\begin{smallmatrix}#1&#2\\#3&#4\end{smallmatrix}}{}

\begin{document}
\title[Symmetries of the finite Heisenberg group]{Symmetries
       of the finite Heisenberg group for composite systems}
\author{M Korbel\'a\v r}
\address{Department of Algebra \\
 Faculty of Mathematics and Physics\\
 Charles University\\ Sokolovsk\'{a} 83, 186 75 Prague 8,
 Czech Republic}
\email{miroslav.korbelar@gmail.com}

\author{J Tolar}
\address{Department of Physics \\
 Faculty of Nuclear Sciences and  Physical  Engineering\\
 Czech Technical University in Prague\\ B\v rehov\'a 7,
 115 19 Prague 1, Czech Republic}
\email{jiri.tolar@fjfi.cvut.cz}


 \begin{abstract}
Symmetries of the finite Heisenberg group represent an important
tool for the study of deeper structure of finite-dimensional quantum
mechanics. As is well known, these symmetries are properly expressed
in terms of certain normalizer. This paper extends previous
investigations to composite quantum systems consisting of two
subsystems --- qudits --- with arbitrary dimensions $n$ and $m$. In
this paper we present detailed descriptions --- in the group of
inner automorphisms of $\GL(nm,\C)$ --- of the normalizer of the
Abelian subgroup generated by tensor products of generalized Pauli
matrices of orders $n$ and $m$. The symmetry group is then given by
the quotient group of the normalizer.
\end{abstract}

\maketitle \vspace{2ex}

\tableofcontents



\section{Introduction}

The special role of the finite Heisenberg group (also known as the
Heisenberg-Pauli group) has been recognized in mathematics as well
as in physics. In the former case, it is closely connected with the
Pauli grading of the classical Lie algebras of the type $\sl(n,\C)$
\cite{PZ88,HPPT02}. In the latter case, its elements provide the
basic quantum observables in the finite-dimensional Hilbert spaces
\cite{Schwinger} of quantum mechanics.

It is then clear that automorphisms or symmetries of the finite
Heisenberg group play very important role in the investigation of
Lie algebras on the one hand \cite{HPPT02,HNPT06,PST06} and of
quantum mechanics in finite dimensions on the other
\cite{Vourdas,SulcTolar07}. These symmetries find proper expression
in the notion of the quotient group of certain normalizer
\cite{PZ89}. The groups of symmetries given by inner automorphisms
were described in \cite{HPPT02} as isomorphic to $\SL(2,\Z_n)$ for
arbitrary $n \in \N$ and as $\Sp(4,\Z_p)$ for $n=p^2$, $p$ prime in
\cite{PST06}.

Note that such symmetry groups are finite-dimensional analogues of
the group $\Sp(2N,\R)$ of linear canonical transformations of the
continuous phase space $\R^{2N}$. Their matrix forms then correspond
to the metaplectic representation of $\Sp(2N,\R)$ in the
infinite-dimensional Hilbert space $L^2(\R^N)$ in which the
Heisenberg Lie algebra is represented \cite{Folland}.

The present investigation is devoted to symmetries of the finite
Heisenberg group for systems composed of two subsystems (qudits) of
arbitrary dimensions $n$, $m$. Such composite systems have been
studied over finite fields, e.g. for both $n=m$ prime \cite{PST06}.
There exist numerous studies of various aspects of the finite
Heisenberg group over finite fields in connection with
finite-dimensional quantum mechanics (FDQM) in Hilbert spaces of
prime or prime power dimensions (see e.g.
\cite{SulcTolar07,Balian,Neuhauser,Vourdas07}).

Our main motivation to study symmetries of the finite Heisenberg
group not in prime or prime power dimensions but for arbitrary
dimensions stems from our previous research where we obtained
results valid for arbitrary dimensions
\cite{HPPT02,StovTolar84,TolarChadz}. Recent paper
\cite{VourdasBanderier} belongs to this direction, too, by dealing
with quantum tomography over modular rings. Also papers
\cite{Digernes,DigernesVV} support our motivation, since they show
that FDQM over growing arbitrary dimensions yields surprisingly good
approximations of ordinary quantum mechanics on the real line.

This paper can be viewed also as a contribution to the study of the
up to now unsolved problem --- the existence or non-existence of the
maximal set of $N+1$ mutually unbiased bases in Hilbert spaces of
arbitrary dimensions $N$. Let us point out that a constructive
existence proof for $N$ prime presented in \cite{SulcTolar07} was
based on consistent use of the symmetry group $\SL(2,\Z_N)$. It is
well known that such maximal sets exist in all prime and prime power
dimensions \cite{WoottersFields89} but the maximal number of
mutually unbiased bases for composite dimensions is unknown as yet.
If known, it would provide a very important contribution to quantum
communication science, where mutually unbiased bases serve as basic
ingredient of secure protocols in quantum cryptography
\cite{Gisin,Alber}. Perhaps this paper may help to unveil the
structure of the maximal set of mutually unbiased bases.

After introductory sections 2, 3 and 4 fixing notation of FDQM,
inner automorphisms and tensor products, in section 5 the group
$\mathcal{G}$ is defined. The central part of the paper is section 6
where the normalizer is completely described and the main Theorem
\ref{37} proved. It contains our principal result that the symmetry
group, being the quotient group of the normalizer, is indeed
isomorphic to $\mathcal{G}$. Moreover, in section 7 the group
$\mathcal{G}$ is described as a suitable extension of groups. More
detailed characterizations of $\mathcal{G}$ are postponed to the
Appendix.

\section{Finite-dimensional quantum mechanics}

Finite-dimensional quantum mechanics (FDQM) \cite{Vourdas,Kibler}
has been developed as quantum mechanics on configuration spaces
given by finite sets equipped with the structure of a finite Abelian
group \cite{StovTolar84}. In the first step one may consider a
single cyclic group $\Z_n$ for given $n\in\N$ as the underlying
configuration space.
\begin{definition}\label{4}
For a given $n\in\N$ set
 $$\omega_{n}:=e^{2\pi \mathrm{i}/n}\in\C.$$
Denote $Q_{n}$ and $P_{n}$ the {\it generalized Pauli matrices} of order $n$,
$$Q_{n}:=\diag(1,\omega_{n},\omega_{n}^{2},\dots,\omega_{n}^{n-1})\in
\GL(n,\C)$$
 and
$$P_{n}\in \GL(n,\C),\ \text{where}\ (P_{n})_{i,j}:=\delta_{i,j-1}, \quad
i,j\in\Z_n.$$
The  $n\times n$ unit matrix will be denoted as $I_{n}$.
 The subgroup of unitary matrices in $\GL(n,\C)$ generated by
$Q_{n}$ and $P_{n}$,
$$\Pi_{n}:=\{ \omega_{n}^j Q_{n}^k P_{n}^l \vert
 j,k,l\in\{0,1,\dots,n-1 \} \}$$
 is called the {\it finite Heisenberg group}.
\end{definition}

The special role of the generalized Pauli matrices has been
confirmed in physical literature as the cornerstone of FDQM, since
their integral powers have physical interpretation of exponentiated
position and momentum operators in position representation
\cite{Schwinger}. As quantum mechanical operators, $Q_{n}$ and
$P_{n}$ act in the $n$-dimensional Hilbert space $\H_n =
\ell^2(\Z_n)$. Further properties of $\Pi_{n}$ are contained in the
following obvious lemma.
\begin{lemma}\label{5}
\begin{enumerate}
\item[(1)] The order of $\Pi_{n}$ is $n^3$.
\item[(2)] The center of $\Pi_{n}$ is
$\{\omega_{n}^j I_n \vert  j\in\{0,1,\dots,n-1 \} \}$.
\item[(3)] $P_{n}Q_{n}=\omega_{n}Q_{n}P_{n}$.
\end{enumerate}
\end{lemma}

We shall also need the following general property of $\Pi_n$.
\begin{lemma}\label{9}
Let $R$ be a $\C$-algebra with unit $1$, $\M_{n}(R)$ be the algebra
of $n\times n$ matrices with entries from $R$ and let $\Pi_{n}$ be
naturally embedded in $\M_{n}(R)$. Let $C\in\M_{n}(R)$ be such that
$P_{n}C=\lambda CP_{n}$ and $Q_{n}C=\mu CQ_{n}$ for some
$\lambda,\mu\in\C\setminus\{0\}$. Then there are $H\in\Pi_{n}$ and
$a\in R$ such that $C=aH $.
\end{lemma}
\begin{proof}
We can suppose that $C\neq 0$. Let
$(Q_{n})_{i,j}=\delta_{i,j}\omega_{n}^{i-1}$ and
$(P_{n})_{i,j}=\delta_{i,j-1}$.  Then
$$(Q_{n}C)_{i,j}=\sum\limits_{k=1}^{n}\delta_{i,k}\omega_{n}^{i-1}C_{k,j}=\omega_{n}^{i-1}C_{i,j}$$
and $$\mu(CQ_{n})_{i,j}=\sum\limits_{k=1}^{n}\mu
C_{i,k}\delta_{k,j}\omega_{n}^{j-1}=\mu\omega_{n}^{j-1}C_{i,j}.
 $$
Hence $(\omega_{n}^{i-j}-\mu)C_{i,j}=0$ for every $i,j$. Since
$C\neq 0$, there are $i_{0}$ and $j_{0}$ such that
$C_{i_{0},j_{0}}\neq 0$. Hence $\mu=\omega_{n}^{i_{0}-j_{0}}$ and
$C_{i,j}=0$ if $i-j\neq i_{0}-j_{0}$ modulo $n$. Thus
$C_{i,j}=\delta_{i,j-j_{0}+i_{0}}c_{i}$ for some
$c_{1},\dots,c_{n}\in R$. Put
$D=\diag(c_{1},\dots,c_{n})\in\M_{n}(R)$. Then, clearly,
$C=P_{n}^{j_{0}-i_{0}}D$ and $D\neq 0$. By assumption,
$P_{n}D=\lambda DP_{n}$. We have
$$(P_{n}D)_{i,j}=\sum\limits_{k=1}^{n}\delta_{i,k-1}\delta_{k, j}c_{j}=\delta_{i,j-1}c_{j}$$ and
$$\lambda(DP_{n})_{i,j}=\sum\limits_{k=1}^{n}\lambda \delta_{i,
k}c_{i}\delta_{k,j-1}=\lambda\delta_{i,j-1}c_{i}.$$ Hence
$c_{j}-\lambda c_{j-1}=0$ for every $j$. It follows that,
$c_{j}=\lambda^{j-1}c_{1}$ and $c_{1}=\lambda^{n}c_{1}$. Since
$D\neq 0$, we have $c_{1}\neq 0$ and $\lambda^{n}=1$. Hence
$\lambda=\omega_{n}^{k}$ and $D=Q_{n}^{k}(c_{1}I_{n})$. Now put
$H=P_{n}^{j_{0}-i_{0}}Q_{n}^{k}\in\Pi_{n}$ and $a=c_{1}\in R$. Then
$C=aH$.
\end{proof}

\section{Inner automorphisms of $\GL(n,\C)$}

For the notion of inner automorphisms, we shall consider (instead of
the Lie algebra $\gl(n,\C)$) the matrix group $\GL(n,\C)$ of all
invertible matrices $n\times n$ over $\C$, since the finite
Heisenberg group was introduced as a subgroup of $\GL(n,\C)$.

\begin{definition}\label{1}
For $M\in\GL(n,\C)$ let $\Ad_{M}\in \Inn(\GL(n,\C))$ be the {\it
inner automorphism} of the group $\GL(n,\C)$ induced by operator
$M\in \GL(n,\C)$, i.e.
 $$\Ad_{M}(X)=MXM^{-1}\, \textrm{for}\, X\in \GL(n,\C).$$
\end{definition}

The following lemma summarizes relevant properties of $\Ad_{M}$.
\begin{lemma}\label{2}
Let $M,N\in \GL(n,\C)$. Then:
\begin{enumerate}
\item[(i)] $\Ad_M \Ad_N=\Ad_{MN}$.
\item[(ii)] $(\Ad_M)^{-1}=\Ad_{M^{-1}}$.
\item[(iii)] $\Ad_{M}=\Ad_{N}$ if and only if  there is a constant
$0\neq\alpha\in\C$ such that $M=\alpha N$.
\end{enumerate}
\end{lemma}

\begin{definition}\label{2a}
We define $\P_n$ as the group
 $$ \P_n=\{\Ad_{Q_n^i P_n^j} \vert (i,j)\in \Z_n \times \Z_n \}.$$
 It is an Abelian subgroup of $\Inn(\GL(n,\C))$ and is generated by
 two commuting automorphisms $\Ad_{Q_n}$, $\Ad_{P_n}$,
 $$ \P_n = \gen{\Ad_{Q_n},\Ad_{P_n}}.$$
\end{definition}
A geometric view is sometimes useful that $\P_n$ is isomorphic to
the \textit{quantum phase space} identified with the Abelian group
$\Z_n \times \Z_n$ \cite{SulcTolar07}.

\section{Tensor products}

According to the well-known rules of quantum mechanics, FDQM on
$\Z_n$ can be extended in a straightforward way to finite direct
products
 $\Z_{m_1}\times\dots\times\Z_{m_f}$.
Here the cyclic groups involved can be viewed as describing
independent quantum degrees of freedom \cite{Schwinger}. The Hilbert
space for FDQM of such a composite system is constructed as the
tensor product
 $\H = \H_{m_1}\otimes\dots\otimes\H_{m_f}$.

In this paper, in order to obtain concrete results, we restrict our
study to special configuration spaces involving just two factors
$\Z_n \times \Z_m$ with arbitrary $n,m\in \N$. Then the
corresponding Hilbert space of FDQM will be $\H = \H_n \otimes \H_m$
of dimension $N=nm$.

We recall the usual properties of the matrix tensor product
$\otimes$. Let $A,A'\in\GL(n,\C)$, $B,B'\in\GL(m,\C)$ and
$\alpha\in\C$. Then:
\begin{enumerate}
\item[(i)] $(A\otimes B)(A'\otimes B')=AA'\otimes BB'$.
\item[(ii)] $\alpha(A\otimes B)=(\alpha A)\otimes B=A\otimes(\alpha B)$.
\item[(iii)] $A\otimes B=I_{nm}$ if and only if
there is $0\neq\alpha\in\C$ such that $A=\alpha I_{n}$ and $B=\alpha^{-1}I_{m}$.
\end{enumerate}

\begin{definition}\label{6}
 Define
$$\P_{n}\otimes \P_{m}:=\set{\Ad_{A\otimes B}}{ A\in\Pi_{n},
B\in\Pi_{m}}\sub \Inn(\GL(nm,\C)).$$
 Further we shall work with generating elements of $\Pi_n \times
 \Pi_m$,
$$A_{1}:=P_{n}\otimes I_{m},\ \ \ A_{2}:=Q_{n}\otimes I_{m},\ \ \
A_{3}:=I_{n}\otimes P_{m},\ \ \ A_{4}:=I_{n}\otimes Q_{m},$$
 and the corresponding inner automorphisms
$$e_{i}:=\Ad_{A_{i}},\ \ \textrm{for}\ \ i=1,2,3,4.$$
\end{definition}
As an easy observation we have the folowing:
\begin{lemma}\label{7}
$\P_{n}\otimes \P_{m}=\prod\limits_{i=1}^{4}\gen{e_{i}}$,
 the direct product of groups $\gen{e_{i}}$ generated by $e_{i}$.
\end{lemma}

\begin{lemma}\label{10}
The centralizer of $\P_{n}\otimes \P_{m}$ in $\Inn(\GL(nm,\C))$ is
equal to  $\P_{n}\otimes \P_{m}$ (i.e.,
$C_{\Inn(\GL(nm,\C))}(\P_{n}\otimes \P_{m})=\P_{n}\otimes \P_{m}$).
\end{lemma}
\begin{proof}
Let $C\in\GL(nm,\C)$ be such that
 $\Ad_{C}\in\Inn(\GL(nm,\C))$ commutes with every element
of $\P_{n}\otimes \P_{m}$.  By \ref{2}(iii), $A_{1}C=\lambda CA_{1}$
and $A_{2}C=\mu CA_{2}$ for some $\lambda,\mu\in\C\setminus\{0\}$.
We can view $C$, $A_{1}$ and $A_{2}$ as elements of
$\M_{n}(\M_{m}(\C))$. By \ref{9}, we get $C=(H\otimes
I_{m})(I_{n}\otimes A)=(H\otimes A)$ for some $H\in\Pi_{n}$ and
$A\in \M_{m}(\C)$. Since $C$ is invertible, we have $A\in\GL(m,\C)$.

Further, there is $0\neq\lambda'\in\C$ such that
$$H\otimes P_{n}A=(I_{n}\otimes P_{m})(H\otimes A)=
\lambda'(H\otimes A)(I_{n}\otimes P_{m})=H\otimes \lambda' AP_{m}.
 $$
Hence $I_{n}\otimes P_{n}A=I_{n}\otimes \lambda' AP_{m}$ and thus
$P_{n}A=\lambda' AP_{m}$. Similarly, $Q_{n}A=\mu' AQ_{m}$ for some
$0\neq\mu'\in\C$. By \ref{9}, we get that $A=\alpha H'$ for some
$H'\in\P_{m}$ and $0\neq\alpha\in\C$. Finally, $C=\alpha (H\otimes
H')$ and $\Ad_{C}\in\P_{n}\otimes \P_{m}$.
\end{proof}

\section{The symmetry group $\mathcal{G}$}

This section is devoted to the definition of the group $\mathcal{G}$
and its principal properties. The proof that $\mathcal{G}$ is indeed
the symmetry group is contained in the next section 6.

We start with the definition of a monoid $\mathcal{S}$, i.e. a set
with a binary associative operation and a  neutral element. Through
this section let $n,m,a,b,d\in\N$ be fixed numbers such that
 $$d\mid
n,\ \ \ d\mid m,\ \ \ n\mid abd\ \ \ \textrm{and}\ \ \ m\mid abd.$$
In order to avoid too complicated notation we will not write out the
dependence of all defined structures on the choice of the numbers
above.

\begin{definition}\label{11}
Put
$$\mathcal{S}:=
\bigg\{\begin{mat2}{A_{11}}{aA_{12}}{bA_{21}}{A_{22}}\end{mat2}\bigg
|\ \ A_{ij}\in\M_{2}(\Z)\bigg\}\sub\M_{4}(\Z).$$
 Further, for
$A=\Big(\begin{smallmatrix}A_{11}&aA_{12}\\bA_{21}&A_{22}\end{smallmatrix}\Big)\in\mathcal{S}$
and
$B=\Big(\begin{smallmatrix}B_{11}&aB_{12}\\bB_{21}&B_{22}\end{smallmatrix}\Big)\in\mathcal{S}$
put
 $$A\equiv B\ \ \ \Leftrightarrow \ \ \ A_{11}
 \equiv_{n}B_{11}\ \ \&\ \ A_{12}\equiv_{d}B_{12}\ \ \&\ \ A_{21}
 \equiv_{d}B_{21}\ \ \&\ \ A_{22}\equiv_{m}B_{22},$$
where $\equiv_k$ is an abbreviation for element-wise `congruent
modulo $k$'.
\end{definition}

\begin{lemma}\label{12}
$\mathcal{S}$ is a submonoid of $(\M_{4}(\Z),\cdot\ )$ with unit
$I_4$ and $\equiv$ is a congruence of $(\mathcal{S},\cdot)$.
\end{lemma}
\begin{proof}
The assertion follows easily from the following equality
 $$\Big(\begin{smat2}{A_{11}}{aA_{12}}{bA_{21}}{A_{22}}\end{smat2}\Big)
 \Big(\begin{smat2}{B_{11}}{aB_{12}}{bB_{21}}{B_{22}}\end{smat2}\Big)=
\Big(\begin{smat2}{A_{11}B_{11}+abA_{12}B_{21}}{a(A_{11}B_{12}+
A_{12}B_{22})}{b(A_{21}B_{11}+A_{22}B_{21})}{abA_{21}B_{12}+
A_{22}B_{22}}\end{smat2}\Big).$$
\end{proof}

\begin{definition}\label{17}
 For
 $A=\Big(\begin{smallmatrix}A_{11}&aA_{12}\\bA_{21}&A_{22}\end{smallmatrix}\Big)
 \in\mathcal{S}$ put $A^{\ast}=
 \Big(\begin{smat2}{A^{T}_{11}}{aA^{T}_{21}}{bA^{T}_{12}}{A^{T}_{22}}\end{smat2}\Big)\in\mathcal{S}$
and denote $[A]$ the class of congruence $\equiv$ containing $A$.
 Further denote
 $$J=\Big(\begin{smat2}{J_{2}}{0}{0}{J_{2}}\end{smat2}\Big)\in\mathcal{S}\
\ \ \textrm{with} \ \
 J_{2}=\Big(\begin{smat2}{0}{1}{-1}{0}\end{smat2}\Big)\in\M_{2}(\Z)
\ \ \textrm{and} \ \ j=[J]\in\mathcal{S}/_{\equiv}$$.
\end{definition}

The following lemma is easy to verify.
\begin{lemma}\label{15}
Let $A,B\in\mathcal{S}$. Then:
\begin{enumerate}
\item[(i)] $(A^{\ast})^{\ast}=A$.
\item[(ii)] $(AB)^{\ast}=B^{\ast}A^{\ast}$.
\item[(iii)] The map $[A]\mapsto [A^{\ast}]$ is a well defined operation on $\mathcal{S}/_{\equiv}$.
\end{enumerate}
\end{lemma}

Now we are prepared to define group $\mathcal{G}$ and state its
simple properties.
\begin{definition}\label{19}
Define
$$\mathcal{G}_{n,m,d,a,b}:=\set{x\in \mathcal{S}/_{\equiv}}{\ \ x^{\ast}jx=j}.$$
Since the entries $n,m,d,a,b$ are fixed we will use a brief notation
and write $\mathcal{G}$ instead of $\mathcal{G}_{n,m,d,a,b}$.
\end{definition}
The following properties are easy to verify.
\begin{lemma}\label{20}
\begin{enumerate}
\item[(i)] $j^{\ast}j=jj^{\ast}=1$ and $j\in \mathcal{G}$.
\item[(ii)] $-1\in\mathcal{G}$, $(-1)j=j^{\ast}$ and $(-1)g=g(-1)$
for every $g\in\mathcal{S}/_{\equiv}$.
\end{enumerate}
\end{lemma}
\begin{lemma}\label{20.1}
 Let $M$ be a finite monoid.
 Let $a\in M$ have a one-sided inverse (left or right).
 Then $a$ is an invertible element.
\end{lemma}
\begin{proof}
Let $a\in M$. Let $b\in M$ be such that $ab=1$.
 Since $M$ is finite, there are $m,n\in\N$, $m>n$ such that $a^{m}=a^{n}$.
 Hence $a^{m-n}=a^{m}b^{n}=a^{n}b^{n}=1$.
 Hence $a(a^{m-n-1})=(a^{m-n-1})a=1$ and $a$ is an invertible element.
The other case being symmetrical.
\end{proof}

\begin{proposition}\label{20.2}
$\mathcal{G}$ is a finite subgroup of the monoid
$\mathcal{S}/_{\equiv}$.
\end{proposition}
\begin{proof}
 Clearly, $\mathcal{S}/_{\equiv}$ is finite. Let $x,y\in\mathcal{G}$.
 Then by \ref{15}(ii) $(xy)^{\ast}j(xy)=y^{\ast}(x^{\ast}jx)y=y^{\ast}jy=j$.
 Hence $xy\in\mathcal{G}$.

 For $x\in\mathcal{G}$ we have $x^{\ast}jx=j$, hence $(j^{\ast}x^{\ast}j)x=1$.
 By \ref{20.1}, $x^{-1}=j^{\ast}x^{\ast}j$.
 Thus $1=xx^{-1}=xj^{\ast}x^{\ast}j$ and $j^{\ast}=xj^{\ast}x^{\ast}$.
 By \ref{20}(ii), $-1\in\mathcal{G}$ and $(-1)j^{\ast}=j$.
 Hence $j=xjx^{\ast}=(x^{\ast})^{\ast}jx^{\ast}$, since $(x^{\ast})^{\ast}=x$.
 It follows that $x^{\ast}\in\mathcal{G}$ and thus
 $x^{-1}=j^{\ast}x^{\ast}j\in\mathcal{G}$. Hence $\mathcal{G}$ is a group.
\end{proof}

As an immediate consequence of the proof of \ref{20.2} we get the
following
\begin{corollary}\label{20.3}
Let $x\in\mathcal{S}/_{\equiv}$.
\begin{enumerate}
\item[(i)] If $x\in\mathcal{G}$ then $x^{-1}=j^{\ast}x^{\ast}j$.
\item[(ii)] $x\in\mathcal{G}$ if and only if $x^{\ast}\in\mathcal{G}$.
\item[(iii)] $x\in\mathcal{G}$ if and only if $xjx^{\ast}=j$.
\end{enumerate}
\end{corollary}

Lengthy derivations of more detailed properties of $\mathcal{G}$ are
postponed to the Appendix. First there is Lemma \ref{23} giving
criterion for a class $[A]$ to belong to $\mathcal{G}$ in terms of
properties of the elements of $A\in\M_4(\Z)$. Second, in Theorem
\ref{30} is proved that $\mathcal{G}$ is a finite group generated by
classes of special elements of $\M_4(\Z)$, see
\begin{quote}\textit{
$\mathcal{G}$ is a group generated by $\{r(1)\}\cup\mathcal{R}$,
 where $r(k)$ and $\mathcal{R}$ are defined in \ref{24}.}
\end{quote}

\section{The normalizer of $\P_{n}\otimes \P_{m}$}

In this central part of the paper the normalizer is completely
described and the main Theorem \ref{37} is proved. It contains our
principal result that the symmetry group, being the quotient group
of the normalizer, is indeed isomorphic to $\mathcal{G}$.

Through this section $n,m\in\N$ are again fixed and we set
$d=\gcd(n,m),\ a=n/d \ {\rm and}\  b=m/d$. With this choice of $d$,
$a$ and $b$ the conditions of the previous section are clearly
fulfilled,
 $$\lcm(n,m)=mn/\gcd(n,m)=ma=nb=abd.$$

\begin{definition}
Define
 $$\mathcal{N}(\P_{n}\otimes
\P_{m}):=N_{\Inn(\GL(nm,\C))}(\P_{n}\otimes \P_{m}),$$
 the normalizer of $\P_{n}\otimes \P_{m}$ in $\Inn(\GL(nm,\C))$.
 Further define
$$\mathcal{N}(\P_{n}):=N_{\Inn(\GL(n,\C))}(\P_{n}),$$
 the normalizer of $\P_{n}$ in $\Inn(\GL(n,\C))$,
 and $$\mathcal{N}(\P_{n})\otimes
\mathcal{N}(\P_{m}):=\set{\Ad_{A\otimes B}}{
A\in\mathcal{N}(\P_{n}), B\in\mathcal{N}(\P_{m})}\sub
\Inn(\GL(nm,\C)).$$
\end{definition}
Clearly, $\mathcal{N}(\P_{n})\otimes \mathcal{N}(\P_{m})\subseteq
\mathcal{N}(\P_{n}\otimes \P_{m})$.
\begin{lemma}\label{31}
 For every $\alpha\in \Aut(\P_{n}\otimes \P_{m})$ there is a unique
 $[A]\in\mathcal{S}/_{\equiv}$ such that
 $$\alpha(e_{j})=\prod\limits^{4}_{i=1}e_{i}^{A_{ij}}\ \ \textit{for each}\ \ j=1,\dots,4.$$
 The map $$\Phi:\Aut(\P_{n}\otimes \P_{m})\to \mathcal{S}/_{\equiv},
 \quad \Phi(\alpha):=[A]$$ is a monoid  monomorphism.
\end{lemma}
\begin{proof}
For every $\alpha\in \Aut(\P_{n}\otimes \P_{m})$ there are
$a_{i,j}\in\Z$ such that
$\alpha(e_{j})=\prod\limits^{4}_{i=1}e_{i}^{a_{ij}}$. Put
$A=(a_{ij})$. The order of both $e_{1}$ and $e_{2}$ is equal to $n$
and the order of both $e_{3}$ and $e_{4}$ is equal to $m$. Hence we
have for $j=1,2$ that
$1=\alpha(e_{j}^{n})=e_{1}^{na_{1j}}e_{2}^{na_{2j}}e_{3}^{na_{3j}}e_{4}^{na_{4j}}$.
Thus $na_{3j}\equiv_{m}0\equiv_{m}na_{4j}$ and
$a_{3j}\equiv_{b}0\equiv_{b}a_{4j}$ for $j=1,2$. Similarly,
$1=\alpha(e_{j}^{m})=e_{1}^{ma_{1j}}e_{2}^{ma_{2j}}e_{3}^{ma_{3j}}e_{4}^{ma_{4j}}$
and $a_{1j}\equiv_{a}0\equiv_{a}a_{2j}$ for $j=3,4$. Hence
$A\in\mathcal{S}$ and using \ref{7} we see that $A$ is unique up to
$\equiv$.

Let $\alpha,\beta\in\Aut(\P_{n}\otimes \P_{m})$. Clearly,
$\Phi(\alpha)=\Phi(\beta)$ implies $\alpha(e_{i})=\beta(e_{i})$ for
every $i$. Hence $\alpha=\beta$. Let $\Phi(\alpha)=[(a_{ij})]$
and $\Phi(\beta)=[(b_{ij})]$. Then
$\alpha(\beta(e_{j}))=\alpha(\prod^{4}_{k=1}e_{k}^{b_{kj}})=\prod^{4}_{k=1}\alpha(e_{k})^{b_{kj}}=
\prod^{4}_{k=1}(\prod^{4}_{i=1}e_{i}^{a_{ik}})^{b_{kj}}=\prod^{4}_{i=1}e_{i}^{c_{ij}}$,
where $c_{ij}=\sum^{4}_{k=1}a_{ik}b_{kj}$. Hence
$\Phi(\alpha\beta)=\Phi(\alpha)\Phi(\beta)$.
\end{proof}

\begin{lemma}\label{32}
Let $$
 \Psi:\mathcal{N}(\P_{n}\otimes \P_{m})\to \Aut(\P_{n}\otimes \P_{m})$$
 $$\Psi(\Ad_{M})(\Ad_{X}):=\Ad_{M}\Ad_{X}\Ad_{M}^{-1}$$
 for every
 $\Ad_{M}\in\mathcal{N}(\P_{n}\otimes \P_{m})$ and
 $\Ad_{X}\in\P_{n}\otimes \P_{m}$.
 Then $\Psi$ is a group homomorphism and $\ker(\Psi)=\P_{n}\otimes \P_{m}$.
\end{lemma}
\begin{proof}
By \ref{10}, $\ker(\Psi)=C_{\Inn(\GL(nm,\C))}(\P_{n}\otimes
\P_{m})=\P_{n}\otimes \P_{m}$.
\end{proof}

\begin{definition}\label{32.1}
Put
$$\lambda_{ij}=e^{w_{ij}2\pi\mathrm{i}/\lcm(n,m)}$$
for $i,j=1,\dots,4$, where $w_{ij}$ are the entries of the matrix
$$W=\left(\begin{array}{cccc} 0&b&0&0\\ -b&0&0&0\\ 0&0&0&a\\ 0&0&-a&0 \end{array}\right).$$
\end{definition}

\begin{lemma}\label{33}
$A_{i}^{k}A_{j}^{l}=\lambda_{ij}^{kl}A_{j}^{l}A_{i}^{k}$ for every $i,j=1,\dots,4$ and $k,l\in\Z$.
\end{lemma}
\begin{proof}
Follows from \ref{5} (and \ref{6}).
\end{proof}

\begin{lemma}\label{34}
$\Phi\Psi(\mathcal{N}(\P_{n}\otimes \P_{m}))\sub\mathcal{G}$.
\end{lemma}
\begin{proof}
Let $\Ad_{G}\in\mathcal{N}(\P_{n}\otimes \P_{m})$, where
$G\in\Inn(\GL(nm,\C)$. By \ref{31}, there is
$A=(a_{ij})\in\mathcal{S}$ such that
$\Phi\Psi(\Ad_{G})=[A]\in\mathcal{S}/_{\equiv}$. Then
$$\Psi(\Ad_{G})(e_{i})=e_{1}^{a_{1i}}e_{2}^{a_{2i}}e_{3}^{a_{3i}}e_{4}^{a_{4i}}$$
for $i=1,\dots,4$. By \ref{2}(iii), there are $0\neq\nu_{i}\in\C$
such that
$$GA_{i}G^{-1}=\nu_{i}A_{1}^{a_{1i}}A_{2}^{a_{2i}}A_{3}^{a_{3i}}A_{4}^{a_{4i}}$$
for $i=1,\dots,4$. Hence
$$GA_{i}A_{j}G^{-1}=GA_{i}G^{-1}GA_{j}G^{-1}=\nu_{i}\nu_{j}A_{1}^{a_{1i}}A_{2}^{a_{2i}}A_{3}^{a_{3i}}A_{4}^{a_{4i}}
A_{1}^{a_{1j}}A_{2}^{a_{2j}}A_{3}^{a_{3j}}A_{4}^{a_{4j}}=$$
$$=\nu_{i}\nu_{j}\lambda_{21}^{a_{1j}a_{2i}}\lambda_{43}^{a_{3j}a_{4i}}A_{1}^{a_{1i}+a_{1j}}A_{2}^{a_{2i}+a_{2j}}A_{3}^{a_{3i}+a_{3j}}A_{4}^{a_{4i}+a_{4j}}$$
by \ref{33}. On the other hand,
$$GA_{i}A_{j}G^{-1}=\lambda_{ij}GA_{j}A_{i}G^{-1}=\nu_{i}\nu_{j}\lambda_{ij}
\lambda_{21}^{a_{1i}a_{2j}}\lambda_{43}^{a_{3i}a_{4j}}A_{1}^{a_{1i}+a_{1j}}A_{2}^{a_{2i}+a_{2j}}A_{3}^{a_{3i}+a_{3j}}A_{4}^{a_{4i}+a_{4j}}.$$
Thus
$\lambda_{21}^{a_{1j}a_{2i}}\lambda_{43}^{a_{3j}a_{4i}}=\lambda_{ij}
\lambda_{21}^{a_{1i}a_{2j}}\lambda_{43}^{a_{3i}a_{4j}}$ for every
$i,j=1,\dots,4$. By \ref{32.1}, this is equivalent to
$$\exp\Big(2\pi\mathrm{i}\frac{b(a_{1i}a_{2j}-a_{1j}a_{2i})+a(a_{3i}a_{4j}-a_{3j}a_{4i})-w_{ij}}{\lcm(n,m)}\Big)=1.$$
This means that $$b\det A^{(i,j)}_{(1,2)}+a\det
A^{(i,j)}_{(3,4)}\equiv_{\lcm(n,m)}w_{ij}$$ for every
$i,j=1,\dots,4$. Hence, by \ref{23}, $[A]\in\mathcal{G}$.
\end{proof}

\begin{lemma}\label{35}
Put
$$R:=\diag(I_{m},Q^{b}_{m},Q^{2b}_{m},\dots,Q_{m}^{(n-1)b})\in\GL(nm,\C).$$
Then $\Ad_{R}\in \mathcal{N}(\P_{n}\otimes \P_{m})$ and
$\Phi\Psi(\Ad_{R})=r(-1)$.
\end{lemma}
\begin{proof}
$R$ is a regular diagonal matrix. Hence $R(Q_{n}\otimes
I_{m})=(Q_{n}\otimes I_{m})R$ and $R(I_{n}\otimes
Q_{m})=(I_{n}\otimes Q_{m})R$.  Further, $$(P_{n}\otimes
I_{m})^{-1}R(P_{n}\otimes
I_{m})=\diag(Q_{m}^{-b},I_{m},Q_{m}^{b},\dots,Q^{(n-2)b}_{m})=(I_{n}\otimes
Q_{m})^{-b}R$$ and $$R(I_{n}\otimes
P_{m})=\diag(P_{m},Q^{b}_{m}P_{m},Q^{2b}_{m}P_{m},\dots,Q_{m}^{(n-1)b}P_{m})=$$
$$=\diag(P_{m},\omega_{m}^{-b}P_{m}Q^{b}_{m},\omega_{m}^{-2b}P_{m}Q^{2b}_{m},\dots,\omega_{m}^{-(n-1)b}P_{m}Q_{m}^{(n-1)b})=$$
$$=\diag(P_{m},(\omega_{n})^{-a}P_{m}Q^{b}_{m},(\omega_{n}^{2})^{-a}P_{m}Q^{2b}_{m},\dots,(\omega_{n}^{n-1})^{-a}P_{m}Q_{m}^{(n-1)b})=
$$$$=(Q_{n}\otimes I_{m})^{-a}(I_{n}\otimes P_{m})R$$ since
$\omega_{m}^{b}=e^{2\pi\mathrm{i}/d}=\omega_{n}^{a}$. Hence
$$\Psi(\Ad_{R})(e_{2})=e_{2}, \, \Psi(\Ad_{R})(e_{4})=e_{4}, \,
\Psi(\Ad_{R})(e_{1})=e_{1}e_{4}^{-a}  \textrm{and}
\Psi(\Ad_{R})(e_{3})=e_{2}^{-b}e_{3}.$$ It follows that
$\Phi\Psi(\Ad_{R})=r(-1)$ where $r(k)$ is defined in \ref{24}.
\end{proof}

For the proof of the following proposition and details see
\cite{HPPT02}.
\begin{proposition}\label{36}
$\Phi\Psi(\mathcal{N}(\P_{n})\otimes
\mathcal{N}(\P_{m}))=\mathcal{R}$ where $\mathcal{R}$ is defined in
\ref{24}.
\end{proposition}

\begin{theorem}\label{37}
\begin{enumerate}
\item[(i)] \begin{equation}
\mathcal{N}(\P_{n}\otimes \P_{m})/\P_{n}\otimes \P_{m}
\cong\mathcal{G}.\end{equation}
\item[(ii)] The group $\mathcal{N}(\P_{n}\otimes \P_{m})$ is generated by
$\mathcal{N}(\P_{n})\otimes\mathcal{N}(\P_{m})$ and $\Ad_{R}$.
\end{enumerate}
\end{theorem}
\begin{proof}
(i) By \ref{30}, $\mathcal{G}$ is generated by $r(1)$ and $\mathcal{R}$.
Hence, by \ref{34}, \ref{35} and \ref{36},
$\Phi\Psi(\mathcal{N}(\P_{n}\otimes \P_{m}))=\mathcal{G}$. Finally,
by \ref{31} and \ref{32}, $\ker(\Phi\Psi)=\P_{n}\otimes \P_{m}$.

(ii) Let $\mathcal{N}$ be a subgroup of $\mathcal{N}(\P_{n}\otimes
\P_{m})$ generated by
$\mathcal{N}(\P_{n})\otimes\mathcal{N}(\P_{m})$ and $\Ad_{R}$. Then
$\ker(\Phi\Psi)=\P_{n}\otimes \P_{m}=\P_{n}\otimes
\P_{m}\sub\mathcal{N}(\P_{n})\otimes\mathcal{N}(\P_{m})\sub\mathcal{N}$
and, by \ref{35}, \ref{36} and \ref{30},
$\Phi\Psi(\mathcal{N})=\mathcal{G}$. Hence
$\mathcal{N}=\mathcal{N}(\P_{n}\otimes \P_{m})$.
\end{proof}

\begin{corollary}\label{38}
$\mathcal{N}(\P_{n}\otimes
\P_{m})=\mathcal{N}(\P_{n})\otimes\mathcal{N}(\P_{m})$ if and only
if $\gcd(n,m)=1$.
\end{corollary}
\begin{proof}
($\Rightarrow$) Let $\mathcal{N}(\P_{n}\otimes
\P_{m})=\mathcal{N}(\P_{n})\otimes\mathcal{N}(\P_{m})$. Then, by
\ref{36} and \ref{35}, $r(1)\in\mathcal{R}$. Hence $a\equiv_{n}0$ and $b\equiv_{m}0$. Thus $n=a$, $m=b$ and
$\gcd(n,m)=\gcd(a,b)=1$.

($\Leftarrow$) If $d=\gcd(n,m)=1$, then $b=m/d=m$ and $R=I_{nm}$, by
\ref{35}. Hence, by \ref{37}(ii), $\mathcal{N}(\P_{n}\otimes
\P_{m})=\mathcal{N}(\P_{n})\otimes\mathcal{N}(\P_{m})$.
\end{proof}

\section{$\mathcal{G}$ as a group extension}

In order to get better insight in the structure of the symmetry
group $\mathcal{G}$ we present it as a suitable extension of groups.
In this section we shall denote it $\mathcal{G}_{n,m}$ equipped with
subscripts $n,m\in \N$.

(i) For $k\in\N$, $A=\Big(\begin{smallmatrix}A_{11}&A_{12}\\
A_{21}&A_{22}\end{smallmatrix}\Big)\in\M_{4}(\Z)$ and
$B=\Big(\begin{smallmatrix}B_{11}&B_{12}\\
B_{21}&B_{22}\end{smallmatrix}\Big)\in\M_{4}(\Z)$ put
$$A\ast_{k}B=\begin{mat2}{A_{11}B_{11}+kA_{12}B_{21}}{A_{11}B_{12}+
A_{12}B_{22}}{A_{21}B_{11}+A_{22}B_{21}}{A_{22}B_{22}+kA_{21}B_{12}}\end{mat2}.$$
Then $\big(\M_{4}(\Z),\ast_{k}\big)$ is a monoid. Indeed, consider
the map
$\nu:\big(\M_{4}(\Z),\ast_{k}\big)\to\big(\M_{4}(\Z),\cdot)$:
$\Big(\begin{smallmatrix}A_{11}&A_{12}\\A_{21}&A_{22}\end{smallmatrix}\Big)\mapsto
\Big(\begin{smallmatrix}A_{11}&kA_{12}\\A_{21}&A_{22}\end{smallmatrix}\Big)$.
Clearly, $\nu$ is injective and $\nu(A\ast_{k}B)=\nu(A)\cdot \nu(B)$
for every $A,B\in\M_{4}(\Z)$. Thus $\nu$ is a monomorphism and
$\ast_{k}$ is associative.

(ii) Let $d\in\N$. By (i), $\ast_{k}$ is a well defined operation
also on $\M_{4}(\Z_{d})$. We denote as in section 5 the class
$j=\Big[\begin{smallmatrix}J_{2}&0\\0&J_{2}\end{smallmatrix}\Big]\in
\M_{4}(\Z_{d})$ and set
$$\mathrm{Sp}_{k}(4,\Z_{d})=\set{x\in\M_{4}(\Z_{d})}{x^{T}\ast_{k}j\ast_{k}x=j}.$$
Using the map $\nu$ we easily get that the monoid $\mathrm{Sp}_{k}(4,\Z_{d})$
is isomorphic to $\mathcal{G}=\mathcal{G}_{d,d,d,k,1}$ (see \ref{19})
and thus it is a group. Note that if $k\equiv_{d}1$ then
$\mathrm{Sp}_{k}(4,\Z_{d})$ is just the usual symplectic group
$\mathrm{Sp}(4,\Z_{d})$.

(iii) Let $n,m\in\N$ and $d=\gcd(n,m)$. Consider now the map
$$\pi:\mathcal{G}_{n,m}\to \mathrm{Sp}_{\frac{nm}{d^{2}}}(4,\Z_{d}) :
\Big[\begin{smallmatrix}A_{11}&aA_{12}\\
bA_{21}&A_{22}\end{smallmatrix}\Big]
\mapsto\Big[\begin{smallmatrix}A_{11}&A_{12}\\
A_{21}&A_{22}\end{smallmatrix}\Big]$$ where $a=n/d$ and $b=m/d$. By
(ii) and \ref{30}, $\pi$ is an epimorphism. Clearly, the kernel of
$\pi$ is $$\ker(\pi)=\bigg\{\left[\begin{array}{cc}I+dA & 0\\ 0 &
I+dB\end{array}\right]\in\mathcal{S}/_{\equiv} \bigg|\det
(I+dA)\equiv_{n}1\ \&\ \det (I+dB)\equiv_{m}1 \bigg\}.$$ Note that
$\det (I+dA)\equiv_{n}1$ is equivalent to
$\Tr(A)+d\det(A)\equiv_{n/d}0$, for a matrix $A\in\M_{2}(\Z)$.
Denote now $$S_{k,l}=\set{[I+kA]\in\M_{2}(\Z_{l})}{A\in\M_{2}(\Z)\
\&\ \Tr(A)+k\det(A)\equiv_{l/k}0}$$ for integers $k,l$ such that
$k\mid l$. Clearly, $S_{k,l}$ is a group and  we get that
$\ker(\pi)=S_{d,n}\times S_{d,m}$. We have thus obtained a short
exact sequence describing an extension of groups:
$$
1 \to (S_{d,n}\times S_{d,m})\to \mathcal{G}_{n,m} \to
\mathrm{Sp}_{\frac{nm}{d^{2}}}(4,\Z_{d}) \to 1.$$

\section{Conclusions}

In this paper complete results are presented concerning the symmetry
group of the finite Heisenberg group of a composite quantum system
consisting of two subsystems with arbitrary dimensions $n,m$. The
corresponding finite Heisenberg group is embedded in $\GL(N,\C)$,
$N=nm$. Via inner automorphisms it induces an Abelian subgroup
$\P_n\otimes\P_m$ in $\Inn{(\GL(N,\C))}$. We have studied the
normalizer of this Abelian subgroup in the group of inner
automorphisms of $\GL(N,\C)$ and have thoroughly described it. The
sought symmetry group $\mathcal{G}$ is the quotient group of the
normalizer (Theorem \ref{37}) and its further characterizations are
given in sections 5, 7 and the Appendix. In section 7 an alternative
description of $\mathcal{G}$ is presented in terms of a group
extension.

The special case of $n=m=p$, $p$ prime, $N=p^2$, is simply observed
to correspond to the group $\mathcal{G}= \Sp(4,\Z_p)$, fully
described in \cite{PST06}. Note that their result is here
generalized to $n=m$ arbitrary (non-prime), leading to $\mathcal{G}=
\Sp(4,\Z_n)$. If $N=nm$, $n,m$ coprime, the symmetry group is,
according to \cite{HPPT02} and Corollary \ref{38}, $\mathcal{G}=
\SL(2,\Z_n)\times \SL(2,\Z_m) \cong \SL(2,\Z_{nm})$

\section*{Acknowledgements}
The first author (M.K.) was supported by the Grant Agency of Charles
University, project \#4183/2009. The second author (J.T.) acknowledges
partial support by the Ministry of Education of Czech Republic,
projects MSM6840770039 and LC06002.

\section{Appendix: Characterization of $\mathcal{G}$}

In the Appendix, the classes $[A]\in\mathcal{G}$ are first
characterized by properties satisfied by elements of matrices
$A\in\mathcal{S}$ (Lemma \ref{23}). Further in Definition \ref{24}
special classes $r(k)$ and set of classes $\mathcal{R}$ in
$\mathcal{S}/_{\equiv}$ are introduced and it is shown that $r(1)$
and $\mathcal{R}$ generate a subgroup of $\mathcal{G}$ (Corollary
\ref{28}). The last part of the Appendix is devoted to the proof
that the set $\{r(1)\}\cup\mathcal{R}$ in fact generates
$\mathcal{G}$ (Theorem \ref{30}). Also recall that $n,m,a,b,d\in\N$
are such that $d\mid n,\ \ \ d\mid m,\ \ \ n\mid abd\ \ \
\textrm{and}\ \ \ m\mid abd$.
\begin{definition}\label{21}
Let $1\leq i<j\leq 4$ and $1\leq k<l\leq 4$.
 For a matrix $A=(a_{rs})_{r,s=1,\dots,4}\in\M_{4}(\Z)$
 let $A^{(i,j)}_{(k,l)}\in\M_{2}(\Z)$ be a submatrix of $A$
 such that $$A^{(i,j)}_{(k,l)}=
 \begin{mat2}{a_{ki}}{a_{kj}}{a_{li}}{a_{lj}}\end{mat2}.$$
For $A=\Big(\begin{smat2}{A_{11}}{aA_{12}}{bA_{21}}{A_{22}}\end{smat2}\Big)\in\mathcal{S}$ denote $\widetilde{A}=\Big(\begin{smat2}{A_{11}}{A_{12}}{A_{21}}{A_{22}}\end{smat2}\Big)$.
\end{definition}

\begin{lemma}\label{23}
Let
$A=\Big(\begin{smat2}{A_{11}}{aA_{12}}{bA_{21}}{A_{22}}\end{smat2}\Big)
\in\mathcal{S}$. Then $[A]\in\mathcal{G}$ if and only if $$\det
\widetilde{A}_{(1,2)}^{(1,2)}+ab\det
\widetilde{A}_{(3,4)}^{(1,2)}\equiv_{n}1,$$
 $$ab\det \widetilde{A}_{(1,2)}^{(3,4)}+\det \widetilde{A}_{(3,4)}^{(3,4)}\equiv_{m}1$$ and
 $$\det \widetilde{A}_{(1,2)}^{(i,j)}+\det \widetilde{A}_{(3,4)}^{(i,j)}\equiv_{d}0$$ for every $1\leq i<j\leq 4$ such that $(i,j)\neq (1,2),(3,4)$.\end{lemma}
\begin{proof}
We have $$A^{\ast}JA=\begin{mat2}{A^{T}_{11}}{aA^{T}_{21}}{bA^{T}_{12}}{A^{T}_{22}}\end{mat2}
\begin{mat2}{J_{2}}{0}{0}{J_{2}}\end{mat2}
\begin{mat2}{A_{11}}{aA_{12}}{bA_{21}}{A_{22}}\end{mat2}=$$
$$=\begin{mat2}{A^{T}_{11}J_{2}A_{11}+abA^{T}_{21}J_{2}A_{21}}{a(A^{T}_{11}J_{2}A_{12}+A^{T}_{21}J_{2}A_{22})}
{b(A^{T}_{12}J_{2}A_{11}+A^{T}_{22}J_{2}A_{21})}{abA^{T}_{12}J_{2}A_{12}+A^{T}_{22}J_{2}A_{22}}\end{mat2}.$$
Hence $[A^{\ast}JA]=[J]$ if and only if $$A^{T}_{11}J_{2}A_{11}+abA_{21}^{T}J_{2}A_{21}\equiv_{n}J_{2},$$
$$abA_{12}^{T}J_{2}A_{12}+A^{T}_{22}J_{2}A_{22}\equiv_{m}J_{2},$$ and
$$A^{T}_{11}J_{2}A_{12}+A_{21}^{T}J_{2}A_{22}\equiv_{d}0.$$ Now use that $A_{11}^{T}J_{2}A_{12}=\Big(\begin{smallmatrix}\det \widetilde{A}_{(1,2)}^{(1,3)}&\det \widetilde{A}_{(1,2)}^{(1,4)}\\\det \widetilde{A}_{(1,2)}^{(2,3)}&\det \widetilde{A}_{(1,2)}^{(2,4)}\end{smallmatrix}\Big)$,
    $A_{21}^{T}J_{2}A_{22}=\Big(\begin{smallmatrix}\det \widetilde{A}_{(3,4)}^{(1,3)}&\det \widetilde{A}_{(3,4)}^{(1,4)}\\\det \widetilde{A}_{(3,4)}^{(2,3)}&\det \widetilde{A}_{(3,4)}^{(2,4)}\end{smallmatrix}\Big)$ and\\  $U^{T}J_{2}U=(\det U) J_{2}$ for $U\in\M_{2}(\Z)$.
\end{proof}

\begin{definition}\label{24}
Let $k\in\Z$. Put $$r(k):=\left[\begin{array}{cccc} 1&0&0&ak\\
0&1&0&0\\ 0&bk&1&0\\ 0&0&0&1
\end{array}\right]\in\mathcal{S}/_{\equiv}$$ and
$$\mathcal{R}:=\Big\{\big[\begin{smallmatrix}S & 0\\ 0 & T\end{smallmatrix}\big]\in\mathcal{S}/_{\equiv} \Big|
\ \ S,T\in\M_{2}(\Z)\ \&\ \det S\equiv_{n}1\ \&\ \det T\equiv_{m}1
\Big\}.$$ Finally, denote
$E_{11}=\big(\begin{smallmatrix}1&0\\0&0\end{smallmatrix}\big)$,
$E_{12}=\big(\begin{smallmatrix}0&1\\0&0\end{smallmatrix}\big)$,
$E_{21}=\big(\begin{smallmatrix}0&0\\1&0\end{smallmatrix}\big)$ and
$E_{22}=\big(\begin{smallmatrix}0&0\\0&1\end{smallmatrix}\big)$
matrices in $\M_{2}(\Z)$.
\end{definition}
Notice that $\mathcal{R}$ is well defined.

\begin{lemma}\label{26}
\begin{enumerate}
\item[(i)] $r(k)=r(1)^{k}$ for every $k\in\Z$ and $r(1)^{d}=1_{\mathcal{G}}$, where $1_{\mathcal{G}}$
is the unit of $I_4$.
\item[(ii)] $r(1)\in\mathcal{G}$.
\item[(iii)] $\mathcal{R}$ is a subgroup of $\mathcal{G}$.
\end{enumerate}
\end{lemma}
\begin{proof}
(i) Clearly, $r(d)=1$. Since $E_{12}E_{12}=0$ we get
$$r(k)r(1)=\begin{matt2}{I_{2}}{akE_{12}}{bkE_{12}}{I_{2}}\end{matt2}
\begin{matt2}{I_{2}}{aE_{12}}{bE_{12}}{I_{2}}\end{matt2}=$$
$$=\begin{matt2}{I_{2}}{a(k+1)E_{12}}{b(k+1)E_{12}}{I_{2}}\end{matt2}=r(k+1).$$ The rest now follows easily by induction.

(ii) Clearly,
$r(1)^{\ast}=\begin{matt2}{I_{2}}{aE_{21}}{bE_{21}}{I_{2}}\end{matt2}$.
We have
$$r(1)^{\ast}jr(1)=\begin{matt2}{I_{2}}{aE_{21}}{bE_{21}}{I_{2}}\end{matt2}
\begin{matt2}{J_{2}}{0}{0}{J_{2}}\end{matt2}
\begin{matt2}{I_{2}}{aE_{12}}{bE_{12}}{I_{2}}\end{matt2}=$$
$$=\begin{matt2}{J_{2}}{aE_{22}}{bE_{22}}{J_{2}}\end{matt2}
\begin{matt2}{I_{2}}{aE_{12}}{bE_{12}}{I_{2}}\end{matt2}=
\begin{matt2}{J_{2}+abE_{22}E_{12}}{a(J_{2}E_{12}+E_{22})}{b(E_{22}+J_{2}E_{12})}{J_{2}+abE_{22}E_{12}}\end{matt2}=j,$$
since $E_{12}J_{2}=E_{22}=-J_{2}E_{12}$ and $E_{22}E_{12}=0$.

(iii) It is enough to show that $\mathcal{R}\sub\mathcal{G}$. For $h\in\mathcal{R}$ we get that
$$h^{\ast}jh=\begin{matt2}{S^{T}}{0}{0}{T^{T}}\end{matt2}
\begin{matt2}{J_{2}}{0}{0}{J_{2}}\end{matt2}
\begin{matt2}{S}{0}{0}{T}\end{matt2}=\begin{matt2}{(\det S)J_{2}}{0}{0}{(\det T)J_{2}}\end{matt2}=j,$$ since $\det S\equiv_{n}1$ and $\det T\equiv_{m}1$.
\end{proof}

\begin{lemma}\label{27}\
\begin{enumerate}
\item[(i)] $\Big(\begin{smat2}{I_{2}}{-aE_{11}}{bE_{22}}{I_{2}}\end{smat2}\Big)=
\Big(\begin{smat2}{I_{2}}{0}{0}{-J_{2}}\end{smat2}\Big)
\Big(\begin{smat2}{I_{2}}{aE_{12}}{bE_{12}}{I_{2}}\end{smat2}\Big)\Big(\begin{smat2}{I_{2}}{0}{0}{J_{2}}\end{smat2}\Big)$.
\item[(ii)] $\Big(\begin{smat2}{I_{2}}{-aE_{22}}{bE_{11}}{I_{2}}\end{smat2}\Big)=
\Big(\begin{smat2}{J_{2}}{0}{0}{I_{2}}\end{smat2}\Big)
\Big(\begin{smat2}{I_{2}}{aE_{12}}{bE_{12}}{I_{2}}\end{smat2}\Big)\Big(\begin{smat2}{-J_{2}}{0}{0}{I_{2}}\end{smat2}\Big).$
\item[(iii)] $\Big(\begin{smat2}{I_{2}}{aE_{21}}{bE_{21}}{I_{2}}\end{smat2}\Big)=
\Big(\begin{smat2}{-J_{2}}{0}{0}{J_{2}}\end{smat2}\Big)
\Big(\begin{smat2}{I_{2}}{aE_{12}}{bE_{12}}{I_{2}}\end{smat2}\Big)\Big(\begin{smat2}{J_{2}}{0}{0}{-J_{2}}\end{smat2}\Big).$
\end{enumerate}
\end{lemma}
\begin{proof}
Easy to verify. Use $-J_{2}=J_{2}^{-1}$.
\end{proof}

\begin{corollary}\label{28}
For a subset $\mathcal{A}\sub\mathcal{G}$ denote $\gen{\mathcal{A}}$
the submonoid of $\mathcal{G}$ generated by $\mathcal{A}$.
Then $\gen{\{r(1)\}\cup\mathcal{R}}$ is a subgroup of $\mathcal{G}$.
\end{corollary}
\begin{proof}
Follows from \ref{26}.
\end{proof}

In the next lemma the entries for a given matrix $X\in\mathcal{S}$,
corresponding to the element $[X]\in\mathcal{G}$, will always be
denoted in the following way:
$$X=\left(\begin{array}{cccc} x_{11}&x_{12}&a\cdot x_{13}&a\cdot
x_{14}\\ x_{21}&x_{22}&a\cdot x_{23}&a\cdot x_{24}\\ b\cdot
x_{31}&b\cdot x_{32}&x_{33}&x_{34}\\ b\cdot x_{41}&b\cdot
x_{42}&x_{43}&x_{44}\end{array}\right)\in\mathcal{S}$$ in order to
simplify the notation.
\begin{lemma}\label{29}
Let $[A]\in\mathcal{G}$, where $A\in\mathcal{S}$. Then:
\begin{enumerate}
\item[(i)] There are $[H]\in\mathcal{R}$ and $[B]\in\mathcal{G}$,
where $H\in\mathcal{S}$, $B\in\mathcal{S}$, such that
$[H][A]=[B]$, $b_{14}=b_{34}=0$ and $b_{44}=\gcd(a_{34},a_{44})$.
\item[(ii)] Let $a_{14}=a_{34}=0$. Then there are
$[H]\in\gen{\{r(1)\}\cup\mathcal{R}}$ and $[B]\in\mathcal{G}$, where
$H\in\mathcal{S}$, $B\in\mathcal{S}$, such that $[H][A]=[B]$,
$b_{14}=b_{34}=0$ and $b_{44}=1$.
\item[(iii)] Let $a_{14}=a_{34}=0$ and $a_{44}=1$.
Then there are $[H]\in\gen{\{r(1)\}\cup\mathcal{R}}$ and
$[B]\in\mathcal{G}$, where $H\in\mathcal{S}$, $B\in\mathcal{S}$,
such that $[H][A]=[B]$, $b_{31}=b_{32}=b_{13}=b_{14}=b_{24}=b_{34}=0$
and $b_{33}=b_{44}=1$.
\item[(iv)] Let $a_{31}=a_{32}=a_{13}=a_{14}=a_{24}=a_{34}=0$ and
$a_{33}=a_{44}=1$. Then $[A]\in\gen{\{r(1)\}\cup\mathcal{R}}$.
\end{enumerate}
\end{lemma}
\begin{proof}
(i) There are $\alpha,\beta\in\Z$ such that $\alpha a_{14}+\beta  a_{24}=\gcd(a_{14},a_{24})=d'$ and $\gamma,\delta\in\Z$ such that $\gamma a_{34}+\delta a_{44}=\gcd(a_{34},a_{44})=d''$. Put
$$H=\left(\begin{array}{cccc} \frac{a_{24}}{d'}&\frac{-a_{14}}{d'}&0&0\\ \alpha&\beta&0&0\\ 0&0&\frac{a_{44}}{d''}&\frac{-a_{34}}{d''}\\ 0&0&\gamma&\delta \end{array}\right).$$ Then $[H]\in\mathcal{R}$. Let $u\in\Z^{4}$ be the last column of $A$. Then
$$Hu=\left(\begin{array}{cccc} \frac{a_{24}}{d'}&\frac{-a_{14}}{d'}&0&0\\ \alpha&\beta&0&0\\ 0&0&\frac{a_{44}}{d''}&\frac{-a_{34}}{d''}\\ 0&0&\gamma&\delta \end{array}\right)
\left(\begin{array}{cccc}a a_{14}\\ a a_{24}\\ a_{34}\\a_{44}\end{array}\right)=
\left(\begin{array}{cccc}0\\ a d'\\ 0\\ d''\end{array}\right).$$
Now put $B=HA$ and the rest is clear.

(ii) Let $a_{14}=a_{34}=0$. By \ref{23}, we have $1\equiv_{m}ab\det \widetilde{A}_{(1,2)}^{(3,4)}+\det \widetilde{A}_{(3,4)}^{(3,4)}=aba_{13}a_{24}+a_{33}a_{44}$, hence $aba_{13}a_{24}+a_{33}a_{44}=1+m\varepsilon$ for some $\varepsilon\in\Z$.
We have $$H'= \left(\begin{array}{cccc} 1&0&0&aa_{13}\\ 0&1&0&0\\
0&ba_{13}&1&a_{33}\\ 0&0&0&1 \end{array}\right)=
\left(\begin{array}{cccc} 1&0&0&aa_{13}\\ 0&1&0&0\\
0&ba_{13}&1&0\\ 0&0&0&1 \end{array}\right)
\left(\begin{array}{cccc} 1&0&0&0\\ 0&1&0&0\\ 0&0&1&a_{33}\\
0&0&0&1 \end{array}\right).$$ Hence, by \ref{26},
$[H']\in\gen{\{r(1)\}\cup\mathcal{R}}$. Let $u\in\Z^{4}$ be the
last column of $A$. Then
$$H'u=\left(\begin{array}{cccc} 1&0&0&aa_{13}\\ 0&1&0&0\\ 0&ba_{13}&1&a_{33}\\ 0&0&0&1 \end{array}\right)\left(\begin{array}{cccc}0\\a a_{24}\\0\\a_{44}\end{array}\right)=
\left(\begin{array}{cccc}aa_{13}a_{44}\\a
a_{14}\\1+m\varepsilon\\a_{44}\end{array}\right).$$ Put $A'=H'A$. Let
$A''\in\mathcal{S}$ be a matrix that differs form $A'$ only on the
position $(3,4)$, where $1$ is instead of $1+m\varepsilon$. Then
$A''\equiv A'$ and thus $[A'']=[H'][A]$. Using (i) for $[A'']$, we
get that there are $H'',B\in\mathcal{S}$ such that
$[H'']\in\mathcal{R}$, $[B]\in\mathcal{G}$, $[H''][A'']=[B]$,
$b_{14}=b_{34}=0$ and $b_{44}=\gcd(1,a_{44})=1$. Now, just put
$H=H''H'$. Then $[H][A]=[H''][H'][A]=[H''][A'']=[B]$. Clearly,
$[H]\in\gen{\{r(1)\}\cup\mathcal{R}}$.

(iii) By \ref{23}, we have $1\equiv_{m}ab\det
\widetilde{A}^{(3,4)}_{(1,2)}+\det
\widetilde{A}_{(3,4)}^{(3,4)}=aba_{13}a_{24}+a_{33}$, hence
$aba_{13}a_{24}+a_{33}=1+m\eta$ for some $\eta\in\Z$. Put
$a'_{33}=a_{33}-m\eta$. Then $$H=
\left(\begin{array}{cccc} a'_{33}&0&-aa_{13}&0\\ 0&1&0&-aa_{24}\\
ba_{24}&0&1&0\\ 0&ba_{13}&0&a'_{33} \end{array}\right)=
\left(\begin{array}{cccc} 1&0&-aa_{13}&0\\ 0&1&0&0\\ 0&0&1&0\\
0&ba_{13}&0&1 \end{array}\right) \left(\begin{array}{cccc}
1&0&0&0\\ 0&1&0&-aa_{24}\\ ba_{24}&0&1&0\\ 0&0&0&1
\end{array}\right)$$ since $aba_{13}a_{24}+a'_{33}=1$. By
\ref{27}(i),(ii) and \ref{26}(i),
$[H]\in\gen{\{r(1)\}\cup\mathcal{R}}$.

Let $A'\in\mathcal{S}$ be a matrix that differs form $A$ only on the position $(3,3)$, where $a'_{33}$ is instead of $a_{33}$. Then $A'\equiv A$. Now, put $B'=HA'$. Let $u$ be the 3-rd column of $A'$ and $v$ be the 4-th column of $A'$. Then $$H\cdot(u,v)=\left(\begin{array}{cccc} a'_{33}&0&-aa_{13}&0\\ 0&1&0&-aa_{24}\\ ba_{24}&0&1&0\\ 0&ba_{13}&0&a'_{33} \end{array}\right)
\left(\begin{array}{cccc} a a_{13}&0\\ a a_{23}&a a_{24}\\a'_{33}&0\\  a_{43}&1\end{array}\right)=
\left(\begin{array}{cccc} 0&0\\ a b'_{23}&0\\1&0\\  b'_{43}&1\end{array}\right)$$
where $b'_{23}=a_{23}-a_{24}a_{43}$ and $b'_{43}=aba_{13}a_{13}+a'_{33}a_{43}$. Clearly, $\det B'^{(1,4)}_{(1,2)}=0$, $\det B'^{(2,4)}_{(1,2)}=0$, $\det B'^{(1,4)}_{(3,4)}=bb'_{31}$ and $\det B'^{(2,4)}_{(3,4)}=bb'_{32}$. Now, using \ref{23}, we get $b'_{31}\equiv_{d}0$ and $b'_{32}\equiv_{d}0$.

Finally, let $B\in\mathcal{S}$ be a matrix that differs form $B'$
only on the position $(3,1)$, where $0$ is instead of $b'_{31}$,
and on the position $(3,2)$, where $0$ is instead of $b'_{32}$. Then $B\equiv B'$. Hence $[B]=[B']\in\mathcal{G}$ and
$[H][A]=[H][A']=[B']=[B]$.

(iv) There are $\nu,\tau\in\Z$ such that $\nu a_{11}+\tau
a_{12}=\gcd(a_{11},a_{12})=d'$. Put $$H=\left(\begin{array}{cccc}
\nu&\frac{-a_{12}}{d'}&0&0\\ \tau&\frac{a_{11}}{d'}&0&0\\
0&0&1&0\\ 0&0&0&1 \end{array}\right).$$ Clearly,
$[H]\in\mathcal{R}$. Put $B=AH$. Then
$$AH=\left(\begin{array}{cccc} a_{11}&a_{12}&0&0\\
a_{21}&a_{22}&aa_{13}&0\\ 0&0&1&0\\ ba_{41}&ba_{42}&a_{43}&1
\end{array}\right) \left(\begin{array}{cccc}
\nu&\frac{-a_{12}}{d'}&0&0\\ \tau&\frac{a_{11}}{d'}&0&0\\
0&0&1&0\\ 0&0&0&1 \end{array}\right)= \left(\begin{array}{cccc}
d'&0&0&0\\ b_{21}&b_{22}&aa_{13}&0\\ 0&0&1&0\\
bb_{41}&bb_{42}&a_{43}&1 \end{array}\right)$$ for some
$b_{ij}\in\Z$.

By \ref{23}, we have $0\equiv_{d}\det \widetilde{B}^{(2,3)}_{(1,2)}+\det
\widetilde{B}^{(2,3)}_{(3,4)}=-b_{42}$ and $0\equiv_{d}\det
\widetilde{B}^{(1,3)}_{(1,2)}+\det \widetilde{B}^{(1,3)}_{(3,4)}=d'a_{13}-b_{41}$.
Hence $b_{42}\equiv_{d}0$ and $b_{41}\equiv_{d}d'a_{13}$. Let
$B'\in\mathcal{S}$ be a matrix that differs form $B$ only on the
position $(4,1)$, where $bd'a_{13}$ is instead of $bb_{41}$, and
on the position $(4,2)$, where $0$ is instead of $bb_{42}$. Then $[B']=[B]\in\mathcal{G}$ and
$[A][H]=[B]=[B']$.

 Further
$$B'=\left(\begin{array}{cccc} d'&0&0&0\\ b_{21}&b_{22}&aa_{13}&0\\ 0&0&1&0\\ bd'a_{13}&0&a_{43}&1 \end{array}\right)=
\left(\begin{array}{cccc} 1&0&0&0\\ 0&1&aa_{13}&0\\ 0&0&1&0\\
ba_{13}&0&0&1 \end{array}\right) \left(\begin{array}{cccc}
d'&0&0&0\\ b_{21}&b_{22}&0&0\\ 0&0&1&0\\ 0&0&a_{43}&1
\end{array}\right).$$ Denote $C$ the first matrix in the
decomposition of $B'$ and $G$ the second one (i.e. $B'=CG$). By
\ref{27}(iii) and  \ref{26}(i),
$[C]\in\gen{\{r(1)\}\cup\mathcal{R}}$. Hence $[C]$ is invertible
in $\mathcal{G}$,  $[G]=[C]^{-1}[B']\in\mathcal{G}$ and it is easy
to see, by \ref{23}, that $[G]\in\mathcal{R}$. Hence
$[B']\in\gen{\{r(1)\}\cup\mathcal{R}}$.

Finally, $[A]=[B'][H]^{-1}\in\gen{\{r(1)\}\cup\mathcal{R}}$.
\end{proof}

\begin{theorem}\label{30}
$\mathcal{G}$ is a group generated by $\{r(1)\}\cup\mathcal{R}$,
where $r(k)$ and $\mathcal{R}$ are defined in \ref{24}.
\end{theorem}
\begin{proof}
Follows immediately from \ref{29}.
\end{proof}



\begin{thebibliography}{99}
\normalsize
\bibitem{PZ88} Patera J and Zassenhaus H 1988
The Pauli matrices in $n$ dimensions and finest gradings of simple
Lie algebras of type $A_{n-1}$ \textit{J. Math. Phys.} {\bf 29}
665--673

\bibitem{HPPT02}
Havl\'{i}\v{c}ek M, Patera J, Pelantov\'a E and Tolar J 2002
Automorphisms of the fine grading of $sl(n,\mathbb{C})$ associated
with the generalized Pauli matrices \textit{J. Math. Phys.} {\bf 43}
1083-1094; arXiv: math-ph/0311015

\bibitem{Weyl}
Weyl H 1931 {\it The Theory of Groups and Quantum Mechanics} (New
York: Dover) pp 272--280

\bibitem{Schwinger}
Schwinger J 1960 Unitary operator bases {\it Proc. Nat. Acad. Sci.
U.S.A.} {\bf 46} 570--579, 1401--1415

\bibitem{HNPT06}
Hrivn\'ak J, Novotn\'y P, Patera J and Tolar J 2006 Graded
contractions of the Pauli graded $\sl(3,\C)$ \textit{Lin. Alg.
Appl.} \textbf{418} 498--550

\bibitem{PST06}
Pelantov\'a E, Svobodov\'a M and Tremblay J 2006 Fine grading of
$\sl(p^{2},\C)$ generated by tensor product of generalized Pauli
matrices and its symmetries \textit{J. Math. Phys.} \textbf{47}
5341--5357

\bibitem{Vourdas} Vourdas A 2004
Quantum systems with finite Hilbert space \textit{Rep. Progr. Phys.}
{\bf 67} 267--320

\bibitem{SulcTolar07} \v{S}ulc P and Tolar J 2007
Group theoretical construction of mutually unbiased bases in Hilbert
spaces of prime dimensions \textit{J. Phys. A: Math. Theor.} {\bf
40} 15099�-15111

\bibitem{PZ89}
Patera J and Zassenhaus H 1989 On Lie gradings I
 \textit{Lin. Alg. Appl.} \textbf{112} 87--159

\bibitem{Folland}
Folland G B 1989 \textit{Harmonic Analysis on Phase Space}
(Princeton, NJ: Princeton University Press)

\bibitem{Balian} Balian R and Itzykson C 1986
Observations sur la m\'ecanique quantique finie \textit{C. R. Acad.
Sci. Paris} {\bf 303} S\'erie I, n. 16, 773--777

\bibitem{Neuhauser} Neuhauser M 2002
An explicit construction of the metaplectic representation over a
finite field \textit{Journal of Lie Theory} {\bf 12} 15--30

\bibitem{Vourdas07} Vourdas A 2007 Quantum systems with finite
Hilbert space: Galois fields in quantum mechanics \textit{J. Phys.
A: Math. Theor.} {\bf 40} R285--R331

\bibitem{StovTolar84}
\v{S}\v{t}ov\'{\i}\v{c}ek P and Tolar J 1984 Quantum mechanics in a
discrete space-time {\it Rep. Math. Phys.} {\bf 20} 157--170

\bibitem{TolarChadz}
Tolar J and Chadzitaskos G 2009 Feynman's path integral and mutually
unbiased bases \textit{J. Phys. A: Math. Theor.} \textbf{42} 245306
(11pp)

\bibitem{VourdasBanderier}
Vourdas A and Banderier C 2010 Symplectic transformations and
quantum tomography in finite quantum systems \textit{J. Phys. A:
Math. Theor.} \textbf{43} 042001 (9pp)

\bibitem{Digernes}
Digernes T, Husstad E and Varadarajan V S 1999 Finite approximation
of Weyl systems \textit{Math. Scand.} {\bf 84} 261--283

\bibitem{DigernesVV}
Digernes T, Varadarajan V S and Varadhan S R S 1994 Finite
approximations to quantum systems \textit{Rev. Math. Phys.} {\bf 6}
621--648

\bibitem{WoottersFields89}
Wootters W K and Fields B D 1989 Optimal state-determination by
mutually unbiased measurements \textit{Ann. Phys. (N.Y.)} {\bf 191}
363--381

\bibitem{Gisin} Gisin N, Ribordy G, Tittel W and Zbinden H
2002 Quantum cryptography \textit{Rev. Mod. Phys.} {\bf 74} 145--195

\bibitem{Alber} Nikolopoulos G M and Alber G 2005
Security bound of two-basis quantum-key-distribution protocols using
qudits \textit{Phys. Rev. A} {\bf 72} 032320

\bibitem{Kibler}
Kibler M R 2008 Variations on a theme of Heisenberg, Pauli and Weyl
\textit{J. Phys. A: Math. Theor.} {\bf 41} 375302

\end{thebibliography}
\end{document}